\documentclass{article}
\usepackage{graphicx} 
\usepackage{latexsym}
\usepackage{color,graphics}
\usepackage{amsmath,amsthm,amssymb,epsfig,graphicx, tikz}
\usepackage{amsfonts}
\usepackage{algpseudocode} 
\usepackage{algorithm}

\def\Z{{\mathbb Z}}

\def\F{{\mathbb F}}
\newcommand{\Ort}{\mathcal{O}}

\DeclareMathOperator{\Soc}{Soc}
\DeclareMathOperator{\Rad}{Rad}
\DeclareMathOperator{\Pker}{Pker}
\DeclareMathOperator{\End}{End}
\DeclareMathOperator{\GL}{GL}

\newcommand\Aut{{\mathrm{Aut}}}

\newtheorem{theorem}{Theorem}

\newtheorem{corollary}[theorem]{Corollary}
\newtheorem{proposition}[theorem]{Proposition}
\newtheorem{definition}{Definition}

\title{Efficient quantum algorithms for 
some instances of
the semidirect discrete logarithm problem}
\author{
Muhammad Imran
\\
Department of Algebra and Geometry,
\\
Budapest University of Technology and Economics,
\\
Egry J\'ozsef u. 1,
H-1111 Budapest, Hungary.
\\
E-mail: \texttt{muh.imran716@gmail.com}
\and
G\'abor Ivanyos
\\
Research Institute for Computer Science and Control,
\\
Hungarian Research Network, 
\\
Kende u. 13-17, H-1111 Budapest, Hungary.
\\
E-mail: \texttt{Gabor.Ivanyos@sztaki.hun-ren.hu}
}

\begin{document}
\maketitle
\begin{abstract}
The semidirect discrete logarithm problem (SDLP) is the
following analogue of the standard discrete logarithm problem in the semidirect product semigroup $G\rtimes \End(G)$ 
for a finite semigroup $G$. 
Given $g\in G, \sigma\in \End(G)$, 
and $h=\prod_{i=0}^{t-1}\sigma^i(g)$ for some integer $t$, the SDLP$(G,\sigma)$, for $g$ and $h$, asks to determine
$t$. As Shor's algorithm crucially depends on commutativity,
it is believed
not to be applicable to the SDLP. 
Previously, the best known algorithm for the SDLP was
 based on Kuperberg's subexponential time quantum algorithm. 
Still, the problem plays a central role in  the security of 
certain proposed cryptosystems in the family of \textit{semidirect product key exchange}. This includes
a recently proposed signature protocol
called SPDH-Sign.  
In this paper, we show that the SDLP is even easier 
in some important special cases. Specifically, for a finite group $G$, we
describe quantum algorithms for the SDLP in 
$G\rtimes \Aut(G)$ for the following two classes of instances:  
the first one is when $G$ is solvable and the second is when
$G$ is a matrix group and a power of
$\sigma$ with a polynomially small exponent is an inner automorphism of $G$.
We further extend the results to groups composed of factors
from these classes. 
A consequence is that SPDH-Sign and similar cryptosystems 
whose security assumption is based on the presumed hardness of the SDLP 
in the cases described above are insecure against quantum attacks.
The quantum ingredients we rely on
are not new: these are
Shor's factoring and discrete logarithm algorithms
and well-known generalizations.
\end{abstract}

\section{Introduction}
  The presumed difficulty of computing discrete logarithm problem (DLP) in certain groups is essential for the security of
the Diffie-Hellman key exchange which is the basis for a number of communication protocols deployed today. However, since the invention of Shor's algorithm \cite{shor1994algorithms}, the problem of computing discrete logarithm can be solved efficiently in the domain of quantum computing. 
  
  Massive efforts have been done in order to construct alternative versions of the discrete logarithm problem that allow for
the Diffie-Hellman key exchange without being vulnerable to Shor's algorithm. 
Since that algorithm takes advantage of the group structure underlying the problem, a DLP 
analogue in the framework of commutative group actions has been proposed. 
It is an instance of a constructive membership testing in 
orbits of commutative permutation groups (on large finite sets),
called \emph{vectorization problem}. 
The framework originally appears in \cite{couveignes2006hard}
and it becomes a central problem of isogeny-based cryptography, CSIDH \cite{castryck2018csidh} for example.
  Another natural approach which is worth consideration to escape from the quantum attack is a DLP analogue in 
non-commutative groups. It is natural in a sense that Shor's algorithm crucially depends on the commutativity of the underlying groups. In this direction, 
an analogue of the DLP in the semidirect product groups has been proposed. The proposal firstly appears in its full generality in \cite{habeeb2013public}. 
  Specifically, let $G$ be a finite semigroup and $\End(G)$ be the monoid of endomorphisms of $G$. Then we have the semidirect product $G\rtimes \End(G)$ where the multiplication is defined by $(g, \sigma)(h, \phi)=(g\sigma(h),\sigma\phi)$. Moreover, we have the formula for exponentiation
\[(g,\sigma)^t=\left(\prod_{i=0}^{t-1}\sigma^i(g), \sigma^t\right),\]
where $\prod_{i=k}^\ell a_i$ stands for the product $a_k\cdot\ldots\cdot
a_\ell$ in $G$.  This leads to an
analogue of the standard 
 discrete logarithm problem  in the semidirect product semigroup defined as follows. Given $g\in G, \sigma\in \End(G),$ and
$h=\prod_{i=0}^{t-1}\sigma^i(g)$ for some integer $t$, determine $t$.

The SDLP is interesting as it allows us to perform a Diffie-Hellman key exchange procedure, known as \textit{semidirect product key exchange} (SPDKE). Suppose two parties, Alice and Bob, agree on a public group $G$, an element $g\in G$, and an endomorphism $\sigma \in \End(G)$. Then they can arrive at the same $G-$element as follows.
\begin{itemize}
    \item[1.] Alice picks a random positive integer $x$ and computes $(g, \sigma)^x=\left(A, \sigma^x\right)$. Then, Alice sends $A=\prod_{i=0}^{x-1}\sigma^i(g)$ to Bob.
   \item[2.] Bob also picks a random  positive integer $y$, computes $(g,\sigma)^y=(B,\sigma^y)$ and sends $B=\prod_{i=0}^{y-1}\sigma^i(g)$ to Alice.
   \item[3.] Alice computes its shared key $K_A=A\sigma^x(B)$.
   \item[4.] Bob computes its shared key $K_B=B\sigma^y(A)$.
\end{itemize}
Note that $K_A=K_B$, as the following calculation shows.
\begin{eqnarray*}
\mbox{~~~~~~~~~}A\sigma^x(B)&=&
\prod_{i=0}^{x-1}\sigma^i(g)\prod_{i=0}^{y-1}\sigma^{x+i}(g)
=\prod_{i=0}^{x+y-1}\sigma^i(g)\\
&=&
\prod_{i=0}^{y-1}\sigma^i(g)\prod_{i=0}^{x-1}\sigma^{y+i}(g)\\
&=& B\sigma^y(A).
\end{eqnarray*}
The key recovery problem of SPDKE is the problem of computing the shared key $K_A=K_B$ from the public information $g, A, B\in G$ and $\sigma\in \End(G)$. Clearly, similar to the case of the standard DLP and the corresponding Diffie-Hellman key exchange, the key recovery problem of SPDKE and the difficulty of SDLP are heavily related. Particularly, if one can solve an instance of the SDLP, then one is also able to break the corresponding SPDKE.

In the description of the SDLP above, an instance of the SDLP in $G\rtimes\End(G)$ is only specified by an endomorphism $\sigma$, hence we can describe the 
SDLP in an alternative, more compact way.

First, we observe some properties of semidirect product semigroups that would be useful for our purpose. 
Let $G$ and $T$ be semigroups and let $\sigma:t\mapsto \sigma_t$ be a 
homomorphism from $T$ to the monoid of endomorphisms of $G$. 
Then the semidirect product $G\rtimes_\sigma T$ is the set $G\times T$ 
equipped with the multiplication $(g,t)(g',t')=(g\sigma_t(g'),tt')$. 
It is straightforward to check that $G\rtimes_\sigma T$ is a semigroup. 
Also, if both $G$ and $T$ are finite groups and $\sigma_1$ is the identity 
map of $G$, then $G\rtimes_\sigma T$ is also a group. There is a natural 
representation $\rho:(g,t)\mapsto \rho_{(g,t)}$ of $G\rtimes_\sigma T$ as a 
semigroup of transformations on $G$, given by $\rho_{(g,t)}(g')=g\sigma_t(g')$. 
This is indeed a representation, i.e., a homomorphism to the semigroup of 
transformations, because we have  $(g,t)(g',t')=(\rho_{(g,t)}(g'),tt')$ 
and \[\rho_{(g,t)(g',t')}=\rho_{(\rho_{(g,t)}(g'),tt')}=\rho_{(g,t)}\circ \rho_{(g',t')}.\] 
If $G\rtimes_\sigma T$ is a group as above then $\rho$ gives a permutation 
representation of the group $G\rtimes_\sigma T$.

Note that 
if $G$ is a monoid and $\sigma$ is a monoid endomorphims of $G$
(that is, $\sigma(1_G)=1_G$), then we have 
$(g,1)^t=(\rho_{(g,1)^t}(1_G),t)$. This shows that,
as already observed by Battarbee \emph{et al.} in \cite{battarbee2022subexponential},
the SDLP can be cast as a
constructive membership problem in an orbit of a transformation semigroup. 
Using the above observation and notations
we have the following definition for the semidirect discrete logarithm that will be used throughout this paper.
\begin{definition}
Let $\sigma$ be an endomorphism of the finite monoid
$G$ with identity element $1_G$ and consider 
the semigroup $G\rtimes_\sigma \Z_{\geq 0}$ where
$\sigma_t=\sigma^t$ for every $t\in \Z_{\geq 0}$. 
Then SDLP$(G,\sigma)$ is the following
problem. Given elements $g$ and $h$ of $G$, determine 
the set of non-negative integers $t$ such that \[h=\rho_{(g,1)^t}(1_G).\]
\end{definition}

The set to be determined is either the empty set, a singleton, or 
$\{x_0+ax:x\in \Z_{\geq 0}\}$
for some integers $x_0\geq 0$ and $a>0$. Indeed, an orbit of a semigroup
generated by a single transformation on a finite set 
consists of a {\em tail} of a certain length called the {\em index},
followed by a recurrent {\em cycle}, whose length is called the {\em
period}. The index can be zero while the period is positive. 
Note that these parameters can be computed by a slight
modification of Shor's period finding quantum algorithm, 
see~\cite{childs2014quantum}. In our case,
the transformation semigroup is generated by $\rho_{(g,1)}$ and our
objective is the orbit of it starting at $1_G$. The solution
set is a singleton if $h$ is in the tail, while in the case when
$h$ is in the cycle, it is $\{x_0+ax:x\in \Z_{\geq 0}\}$, where $a$
is the period and $x_0$ is the position of $h$ in the cycle, shifted by
the index.

We remark that the assumptions that $G$ is a monoid and that $\sigma$ is a monoid
endomorphism of $G$ 
are rather technical, though they offer some
notational conveniences. In the general semigroup case, one should
solve the equation 
$h=\rho_{(g,1)^{t-1}}(g)$.

Battarbee \emph{et al.} \cite{battarbee2022subexponential} present a subexponential quantum algorithm for the SDLP in
so-called the \textit{easy} family of semigroups $\{G_p\}_{p\in P}$ 
for some countable set $P$. A family of
semigroups $\{G_p\}_{p\in P}$ is called easy if the size $|G_p|$ 
grows monotonically and polynomial in $p$, 
and the evaluation costs of $gh$ and $\sigma(g)$ is $\Ort((\log p)^2)$ 
for any $p\in P$, $g,h\in G_p$, and $\sigma \in \End(G_p)$. 
Indeed, the critical problem is determining the position of $h$
in the cycle, which is actually an instance of the vectorization 
problem, and hence reduces 
to the abelian hidden shift problem  for which Kuperberg's subexponential 
time algorithm \cite{kuperberg2005subexponential} is available.
On the other hand, there exist several efficient algorithms that 
break the SPDKE protocols in some specific groups without solving the corresponding SDLP, instead exploiting the structure of the platform groups to directly solve the corresponding key recovery problem. See \cite{battarbee2022semidirect} for
a more detailed survey on the semidirect product key exchange. The most recent work in this direction is by Battarbee \emph{et al.} \cite{Battarbee2023aspdh}. They propose a post-quantum signature scheme, called SPDH-Sign, where the security depends on the presumed difficulty of the group case of the SDLP. Moreover, they propose non-abelian groups of order $p^3$ for some odd prime $p$ as candidate groups for SPDH-Sign.

In this paper, we work over black-box groups 
with non-necessarily unique encoding of elements to obtain 
sufficiently general results. (Together with assuming
ability of evaluating powers of $\sigma$, this
corresponds to the easy families 
of  \cite{battarbee2022subexponential}.)
The concept of black-box groups 
was introduced by Babai and Szemer\'edi \cite{babai1984complexity} 
for studying the structure of finite matrix groups. 
Elements of a black-box group $G$ are represented by binary strings 
of a certain length and the group itself is given by a list
of generators. The group operations are given by oracles.
Here we also assume an oracle for 
computing $\sigma^j(g)$ for $g\in G$ and $j\in \Z_{>0}$. 
In general, it is not required that every group element
is represented by a unique code-word. Instead, there is also an 
oracle for testing whether two strings represent the same group element. 
Here we assume a stronger oracle, a {\em labeling}. It is a 
function $\lambda$ defined on the code-words for the group 
elements where $x$ and $y$ represent
the same group element if and only if $\lambda(x)=\lambda(y)$.
We use the term {\em black-box group with
unique labeling} for that sort of black-box groups.
The labeling makes it possible to compute the structure
of $G$ when $G$ is a solvable black-box group
by the quantum algorithm of \cite[Theorem~7]{ivanyos2001efficient}. (In that
paper the term {\em secondary encoding} is used for the labeling.)
The notion includes black-box groups with unique encoding. We need the
generalization in order to handle certain factor groups. 
To illustrate how this can occur, assume
that initially we work with a matrix group $G$ and $\sigma$ is
given as conjugation by a matrix (possibly outside $G$)
and we have another, non-faithful matrix representation $\phi$ of $G$ 
whose kernel is $\sigma$-invariant. Suppose further that we need 
to solve the SDLP for $\phi(g)$ and $\phi(h)$ in $\mbox{Im}(\phi)$ 
and the automorphism induced by $\sigma$. (Recall that this is the unique map 
${\overline \sigma}:\mbox{Im}(\phi)\rightarrow\mbox{Im}(\phi)$
satisfying $\psi(\sigma(x))={\overline \sigma}(x)$. It is well-defined
as the kernel of $\phi$ is required to be $\sigma$-invariant.)
It turns out that we would 
have difficulties with evaluating powers of the induced automorphism 
if we used the natural unique encoding of the 
elements of $\mbox{Im}(\phi)$ by matrices. (In general,
this would require finding
finding an element of the pre-image $\phi^{-1}(x)$ for $\mbox{Im}(\phi)$.)
We get around the issue by using the original
matrices to encode the elements of $\mbox{Im}(\phi)$ and to multiply them;
while considering $\phi$ as a labeling (and possibly also as further help). 
This gives us a simple way to evaluate the induced automorphism.

\medskip

The SDLP$(G, \sigma)$ is called the \textit{group-base} case if $G$ is a group, 
and we call it the \textit{(full) group} case when $G$ is a group and 
$\sigma$ is an automorphism of $G$. In this paper we focus on the 
group-base case. If, in addition, $\sigma$ is an automorphism of $G$
then one could replace the monoid $\Z_{\geq 0}$ with an appropriate
finite cyclic group $\Z_m=\Z/m\Z$ where $m$ is a multiple of the order
of $\sigma$ and work over the finite semidirect product group
of $G$ and $\Z_m$. This justifies the terminology.

\paragraph{Contributions.} In this paper, we provide an analysis of the SDLP 
in some interesting classes of groups. Particularly, in section \ref{sec2}, 
we first give a reduction from the group-base case to the group case of the SDLP. 
Moreover, using essentially the same idea, we show that there exists a recursion from the SDLP in a group into its quotient groups and subgroups. In section \ref{sec3}, we then propose efficient quantum algorithms based on Shor's algorithm for the group case SDLP$(G, \sigma)$ for the following cases: 
\begin{itemize}
    \item[1.] The automorphism $\sigma$ is of small order, i.e., polynomial in $\log|G|$;
    \item[2.] The group $G$ is solvable;
    \item[3.] The group $G$ is a matrix group over a finite field, 
	i.e., $G\leq \GL_d(\F_q)$, where $q$ is a power of a prime and
		$\sigma$ is an inner automorphism of $G$; 
    \item[4.] A flag $1=M_0<M_1<\ldots<M_k=G$ of $\sigma$-invariant 
normal subgroups $M_i\lhd G$ is given together with homomorphisms $\psi_i$
from $M_i$ with kernel $M_{i-1}$ ($i=1,\ldots,k$), where for each $i$,
$\psi_i$ maps $M_i$ to either
\begin{itemize} 
\item[4.1]
a black-box group with unique labeling and when automorphism of
$\mbox{Im}(\psi_i)$ induced by $\sigma$ has polynomially small order; or
\item[4.2]
a solvable black-box group with unique labeling; or 
\item[4.3] a matrix group 
over a finite field, in which case we also assume that a power 
of the induced automorphism  
with a polynomially small exponent coincides with the conjugation 
by some matrix.
\end{itemize}
\end{itemize}

As a consequence, SPDH-Sign protocol in \cite{Battarbee2023aspdh} and 
all other SPDKE cryptographic protocols whose platform groups are 
in the above cases  do not belong to the realm of post-quantum cryptography. 
We remark that, a normal series together with the homomorphisms having the properties
required in item~4.,~can be efficiently computed for quite a wide class 
of finite groups using advanced algorithms of computational group theory. 
These include matrix groups over finite fields of odd characteristic
making the innerness assumption of item~3.~unnecessary when $q$ is odd,
see the Appendix for a sketch of proof. We even think that it is difficult
to propose any "concrete" platform group that item 4.~is not applicable
to, so a viable platform for SPDH-Sign protocol should be a 
semigroup quite far from any group.

\section{Reduction and recursion of SDLP}\label{sec2}
In this section, we provide the reduction of the group-base case to the group case, and we also describe a recursion tool that passes the SDLP in a group to its quotient groups and subgroups. 

From $(g,1)^t=(\rho_{(g,1)^t}(1_G),t)$ we infer the following identity
\begin{equation}
    \rho_{(g,1)^{rt}}(1_G)=\rho_{(\rho_{(g,1)^r}(1_G),r)^t}(1_G).
\end{equation}
We will frequently use this fact to reduce an instance of the SDLP for the endomorphism
$\sigma$ to an instance for $\sigma^r$ in place of $\sigma$ with suitable
choices of $r$.

\subsection{Reduction from the group-base case to the group case}
Let $G$ be a finite group and $\sigma$ be an endomorphism of $G$. We will describe a reduction from SDLP$(G,\sigma)$ to SDLP$(K, \sigma')$ where $K$ is a subgroup of $G$ and $\sigma'$ is the restriction of $\sigma$ to $K$ which forms an automorphism.

Let $K=\cap_{t=0}^\infty \sigma^t(G)$ and let $k_0$ be the smallest non-negative integer such that $K=\sigma^{k_0}(G)$. 
Obviously, $k_0\leq \lceil \log|G| \rceil $. Let $k\geq k_0$, where such a $k$ can be "blindly" chosen by taking an integer greater than 
a known upper bound for $\log|G|$. (Such an upper bound can be $\ell$, where
binary strings of length $\ell$ encode the group elements.)
Then $K=\sigma^k(G)$ and the restriction of $\sigma$ to $K$ is an automorphism of $K$. Let $r$ be the length of the orbit $\{\rho_{(\sigma^k(g),1)^t}(1_G):t\in \Z_{\geq 0}\}$
and put $M=\ker \sigma^{k}=\ker \sigma^{k_0}$. 
Then $K\cong G/M$, $K\cap M=\{1_G\}$, and we have 
\[r=\min \{t\in \Z_{> 0}:\rho_{(\sigma^k(g),1)^t}(1_G)=1_G\}=
\min \{t\in \Z_{> 0}:\rho_{(g,1)^t}(1_G)\in M\}.\]
Let $g'=\rho_{(g,1)^r}(1_G)$. Then $\rho_{(g,1)^{rt}}(1_G)=\rho_{(g',r)^t}(1_G)$. As $g'\in M=\ker\sigma^{k}$, 
for $rt\geq k$ we have $\sigma^{rt}(g')=1_G$, and hence 
$\rho_{(g',r)^{t+1}}(1_G)=g'\sigma^r(g')\cdot\ldots\cdot\sigma^{rt}(g')\sigma^{rt}(g')=
g'\sigma^r(g')\cdot\ldots\cdot\sigma^{r(t-1)}(g')=
\rho_{(g',r)^{t}}(1_G)$.
It follows that 
\begin{equation}\label{eq2}
    \rho_{(g,1)^{r(t+1)+s}}(1_G)= \rho_{(g,1)^{rt+s}}(1_G),
\end{equation}

By equation (\ref{eq2}), 
   if the solution set of the SDLP in $K$ for $\sigma^k(g)$ and $\sigma^k(h)$ is $\{s+rt:t\in \Z_{\geq 0}\}$ for some $0\leq s<r$, then the set of solutions of the SDLP in $G$ for $g$ and $h$ is either the empty set, a singleton $\{s+rt_0\}$, 
or $\{s+rt:t\in \Z_{\geq t_0}\}$, for some $t_0\leq \lceil k_0/r\rceil\leq \lceil\log|G|\rceil$. 
 Therefore, one can solve the SDLP$(G,\sigma)$ for $g$ and $h$ by solving SDLP$(K, \sigma|_K)$ for $\sigma^{k}(g)$ and $\sigma^k(h)$, 
followed by an exhaustive search. This gives the following theorem.

\begin{theorem}
\label{thm:toautocase}
    There is a classical polynomial time reduction from an instance of the group-base case SDLP to an instance of the group case SDLP.
\end{theorem}

\subsection{An easy reduction}

In the group case, we have the following simple reduction based on 
brute force. This will be useful when a power of the automorphism
$\sigma$ with polynomially small exponent has some desired property.

\begin{proposition}
\label{prop:topowerofaut}
Assume that $\sigma$ is an automorphism of the group $G$.
Then, for every positive integer $k$, 
SDLP$(G,\sigma)$ can be reduced to $k$ instances of 
SDLP$(G,\sigma^{k})$. 
\end{proposition}

\begin{proof}
We look for the smallest
non-negative solution of the SDLP in the form 
$s+tk$ for $s=0,\ldots,k-1$. We have 
\begin{equation*}
    \begin{split}
        \rho_{(g,1)^{s+tk}}(1_G) & =\rho_{(g,1)^{s}}(\rho_{(g,1)^{tk}}(1_G))\\
        & =\rho_{(g,1)^{s}}(\rho_{(g,1)^{k}}^t(1_G)))\\
        & =\rho_{(g,1)^{s}}(\rho_{(\rho_{(g,1)^{k}}(1_G),k)}^t(1_G))),
    \end{split}
\end{equation*}
whence $h=\rho_{(g,1)^{s+tk}}(1_G)$ 
if and only if $\rho_{(g,1)^{-s}}(h)=\rho_{(\rho_{(g,1)^{k}}(1_G),k)^t}(1_G))$.
Let $g'=\rho_{(g,1)^{k}}(1_G)$ and $h'=\rho_{(g,1)^{-s}}(h)$.
Then, we need to solve the SDLP for $g'$ and $h'$, where we replace
$\sigma$ by $\sigma^{k}$.
\end{proof}

\subsection{Recursion into quotient groups and subgroups}\label{recursion}
We will show that one can solve the SDLP$(G,\sigma)$, for a group $G$ and $\sigma\in \Aut(G)$, by recursively solving an instance of the SDLP in a quotient group and a subgroup of $G$. 
The main idea of recursion is essentially the same 
as those used in the preceding subsections.

\begin{theorem}
\label{thm:recurs}
Let $G$ and $\overline G$ be black-box groups with 
unique labeling and let an
automorphism $\sigma$ of $G$ be given by a black box
for evaluating the powers $\sigma^i$ on codewords for
group elements. Assume that we are given
a $\sigma$-invariant normal subgroup $M$ of $G$
and a group homomorphism $\psi:G\rightarrow \overline G$ 
with kernel $M$. We assume that $\psi$ can be evaluated efficiently
and we have a black box for evaluating
powers of the automorphism $\overline \sigma$ of $\mbox{Im}(\psi)$
induced by $\sigma$. Then SDLP$(G,\sigma)$
can be reduced to an instance of 
SDLP$(\mbox{Im}(\psi),{\overline \sigma})$ and
an instance of SDLP$(M,\sigma_{|M}^{n_0})$ 
for some integer $n_0$. 
\end{theorem}

\begin{proof}
Every solution of SDLP$(G,\sigma)$ for $g$ and $h$ is a solution of 
the SDLP$(\mbox{Im}(\psi),\overline{\sigma})$ for $\psi(g)$ and $\psi(h)$.
If there is no solution for the problem in $\mbox{Im}(\psi)$, 
then there is no solution for the problem in $G$ either. 

Otherwise, the set of solutions in 
$\mbox{Im}(\psi)$ is the residue class $\{t_0+n_0t\}$ for 
some $0\leq t_0< n_0$, where $n_0=|\{\rho_{(\psi(g),1)^t}(1_G):t\in \Z\}|$. 
Note that $n_0$ is the
smallest positive integer such that $\rho_{(g,1)^{n_0}}(M)=M$.
We have $g'=\rho_{(g,1)^{n_0}}(1_G)\in M$ and also 
\[h'=\rho_{(g,1)^{-t_0}}(h)=(\rho_{(g,1)^{-t_0}}\circ\rho_{(g,1)^{t_0+n_0t}})(1_G)=\rho_{(g,1)^{n_0t}}(1_G)\in M.\] 
This gives that
    the solutions of 
SDLP$(G,\sigma)$ for $g$ and $h$ are exactly the numbers of 
the form $t_0+n_0t$, where $t$ is a solution of 
SDLP$(M, \sigma^{n_0})$ for $g'$ and $h'$. 
\end{proof}

By considering the equivalent "backward" version of the SDLP, that is, solving
$1_G=\rho_{(g,1)^{t}}(h)$, the recursion suggested by the proof of the
theorem can be interpreted as driving first to $M$ by solving the SDLP
in $\mbox{Im}(\psi)\cong G/M$ and then, inside $M$, driving further
to the identity element.

A general straightforward way to evaluate 
the induced automorphism (and its powers) is 
based on computing an arbitrary element of the pre-image 
$\psi^{-1}({\overline x})$ for each $x\in\mbox{Im}(\psi)$.
This can be facilitated by replacing $\overline G$ with
the black-box group $H$ encoded by pairs $(x,\psi(x))$,
where $x$ is a code-word for an element of $G$. For multiplication
we use the oracle for $G$ and re-evaluate $\psi$ on the product.
For labeling, we use the labeling of $\overline G$. Of course,
there are many cases when this trick can be replaced by a
simple direct method for evaluating $\overline \sigma$.
This holds in particular when $\overline G=\Z_p^d$ with
the standard representation by column vectors modulo $p$.

 \section{Quantum algorithms for the group case SDLP}\label{sec3}
In this section, we will prove the following main result of the paper.

\begin{theorem}\label{thm:main}
    Let $G$ be a group and $\sigma\in \Aut(G)$. 
We assume that $G$ is a black-box group with a unique labeling 
of elements and we also have a black box for computing $\sigma^i(g)$
($i\in \Z{\geq 0},g\in G$). 
Suppose that we are given 
a series $1=M_0<M_1<\ldots M_k=G$ of $\sigma$-invariant normal
subgroups $M_i\lhd G$ together with homomorphisms $\psi_i:M_i\rightarrow
 {\overline G}_i$ ($i=1,\ldots,k$) with kernel $M_{i-1}$ ($i=1,\ldots,k$).
Let ${\overline \sigma}_i$ denote the automorphism of $\mbox{Im}(\psi_i)$
induced by $\sigma_{|M_i}$.
Assume further that, for each $i$, either
\begin{enumerate}
\item[(0)] $\mbox{Im}(\psi_i)$ is of polynomial size; or
\item[(1)] ${\overline \sigma}_i$ has polynomial order; or
\item[(2)] $\mbox{Im}(\psi_i)$ is solvable;
\item[(3)] ${\overline G}_i\leq \GL_{d_i}(\F_{q_i})$ for some
positive integer $d_i$ and for some prime power $q_i$, moreover,
there exists a polynomially bounded integer $n_i$ and a matrix
$a_i\in \GL_{d_i}(\F_{q_i})$ such that
${\overline \sigma}_i^{n_i}(x)=a_i^{-1}xa_i$ for every 
$x\in \mbox{Im}(\psi_i)$. 
\end{enumerate}
For items (0), (1) and (2), we assume that ${\overline G_i}$
is a black-box group with unique labeling. For item (4), neither
$n_i$ nor $a_i$ are assumed to be given, their mere existence
is sufficient.
(By "polynomial" we mean polynomial in the 
maximum of the lengths
of the bit strings used for encoding and labeling the elements
of the groups $G$ and ${\overline G}_i$ ($i=1,\ldots,k$).
Then SDLP$(G,\sigma)$ can be solved in quantum polynomial time. 
\end{theorem}

When $k=1$, condition of type (0) means that 
$G$ itself is of polynomial size, that of type
(1) means that $\sigma$ itself 
has polynomially small order, that of type (2) means that 
$G$ is solvable. The standard descriptions of 
simple groups of Lie type define them as factors of certain
matrix groups over finite field. The quotient is taken to be the center
of the matrix group, so the simple group has a representation as a matrix group
by the conjugation action on the matrix algebra spanned by the 
covering matrix group. Also,
the outer automorphism group of a finite simple group is of
polynomial size. Therefore, these groups are covered by conditions
of type (3). 

The algorithm for polynomially small groups is the
straightforward trial and error. In the first three subsections of this section we
give efficient algorithms for groups/automorphisms satisfying conditions (1),
(2), or (3). In the fourth subsection we show how to use these ingredients
and Theorem~\ref{thm:recurs} to prove Theorem~\ref{thm:main}. 

Note that the order of $\sigma$ can be computed in quantum 
polynomial time using Shor's period finding method applied to
the functions $t\mapsto \sigma^t(x_i)$ for the generators $x_i$
of the group $G$ and taking the least common multiple of these periods.
The order can be factorized using Shor's factoring
algorithm. The length of the orbit $\{\rho_{(g,1)^t}(1_G):t\in \Z\}$
can be determined and factorized in a similar way. Based on
these observations, in the algorithms below
we assume that these numbers are already computed and factorized.
The solution set is either empty or the the residue class 
of an arbitrary solution
modulo the period. So it is sufficient to find any solution,
e.g., the smallest non-negative one.

\subsection{The SDLP for small order automorphisms}
\label{subsec:smallorderaut}

In this subsection we prove the following result.

\begin{proposition}
\label{prop:smallorderaut}
Let $G$ be a black-box group with unique labeling. Then
SDLP$(G,\sigma)$ can be solved by a quantum algorithm 
in time polynomial in the order of $\sigma$ and the 
length of the code-words together with the labels
of the group elements.
\end{proposition}

\begin{proof}
By Proposition~\ref{prop:topowerofaut}, it is sufficient to
prove the case when $\sigma$ is trivial. Then
$\rho_{(g,1)}(x)=gx$, whence $\rho_{(g,1)^t}(1_G)=g^t$ for every integer $t$.
Thus, solving the SDLP for $g$ and $h$ is the same as computing the 
base-$g$ discrete logarithm of $h$, which can be accomplished by
Shor's algorithm.
\end{proof}

\subsection{The SDLP in solvable groups}

In this part, we first present a quantum algorithm for 
the SDLP on elementary abelian groups. We then show
how Theorem~\ref{thm:recurs} can be used to reduce
the general solvable case to instances of 
the elementary abelian case. 

\begin{theorem}
Let $G=\Z_p^d$, the (additive) group of column vectors
of length $d$ over the integers modulo $p$, where $p$
is a prime number and let $\sigma$ be an automorphism
of $G$, given as a $d$ by $d$ non-singular matrix.
Then SDLP$(G,\sigma)$ can be solved by a quantum algorithm
in time polynomial
in $\log p$ and $d$.
\end{theorem}

\begin{proof}
We consider $G$ as a vector space of dimension $d$ over
the finite field $\Z_p$. We take a minimal nontrivial 
$\sigma$-invariant subspace $M$ of $G$. This can be done,
e.g., by a classical randomized method
based on computing the rational Jordan normal form
of $\sigma$, see \cite{giesbrecht1995nearly}.
Then the factor space $M$ has no proper 
nontrivial $\sigma$-invariant subspace. Iterating this in $G/M$,
we eventually obtain a flag of subspaces $(0)=M_0<M_1<\ldots<M_k=G$
such that there is no $\sigma$-invariant subspace strictly
between $M_{i-1}$ and $M_i$. Then, by Theorem~\ref{thm:recurs},
the problem is reduced to the case when 
$G$ has no proper nontrivial $\sigma$-invariant subspace. 
Suppose that we have an instance of that case. 

If $1_G$ is an eigenvalue of $\sigma$ then $d=1$ and $\sigma$ is trivial. 
It follows that $\rho_{(g,1)^t}(1_G)=g^t$. 
Using the additive notation for $\Z/p\Z$, we need to solve $h=t\cdot g$. 
If $g=0$ and $h=0$ then every integer is a solution, while if $g=0$ and 
$h\neq 0$ then there is no solution. 
If $g\neq 0$ let $g'$ stand for the multiplicative inverse of 
$g$ in the field $\Z/p\Z$. Then the solutions are $\{hg'+tp:t\in \Z\}$. 
(Actually, the case when $\sigma$ is trivial is a special case 
of the broader case already discussed in Subsection~\ref{subsec:smallorderaut}.)

If $1_G$ is not an eigenvalue then we do the following. We compute the matrix $B$ of $\sigma$ in the standard basis of $(\Z/p\Z)^d$. Then, using again the additive notation, we can write $\rho_{(g,1)^t}(1_G)$ as \[\sum_{j=0}^{t-1}B^jg=(B^t-I)(B-I)^{-1}g.\]
Let $\cal B$ be the matrix algebra generated by $B$. Then $\cal B$ is isomorphic to the field $\F_{p^d}$ and the action of $B$ on $(\Z/p\Z)^d$ can be identified with the multiplication by the field element $B$ on the additive group of the field. Thus we can view $g$ and $B$ as field elements. The solutions of the SDLP can be obtained by solving a discrete log problem in this field by computing 
base-$B$ logarithm of $I+(B-I)g^{-1}h$. This can be done by Shor's quantum
algorithm.
\end{proof}

The method for the elementary abelian case, in combination with the recursion 
tool (Theorem~\ref{thm:recurs}), gives an efficient quantum algorithm for solving the SDLP in solvable 
groups. More precisely, we obtain the following result.

\begin{theorem}
\label{thm:solvable}
Assume that $G$ is a solvable black-box group with unique labeling.
Then SDLP$(G,\sigma)$ can be solved in quantum polynomial time.
\end{theorem}

\begin{proof}

Using the labeling,
by \cite[Theorem~7]{ivanyos2001efficient} 
which is based on the Beals-Babai algorithm \cite{babai1999polynomial}, 
we can compute a composition series of $G$ with explicit isomorphisms
between the composition factors and additive groups $\Z_p$ for various 
primes $p$. In particular, we obtain a maximal normal subgroup $N$
of $G$ together with a homomorphism $\phi:G\mapsto \Z_p$. For any
positive integer $j$, let $N_j=\cap_{i=0}^{j-1}\sigma^{i}(N)$. Note
that $N_{j+1}=N_j\cap \sigma^{j}(N)$ and if $N_{j+1}=N_j$ then
$N_{j'}=N_j$ for any integer $j'>j$ and $N_j$ is $\sigma$-invariant. 
This equality happens for an integer $j$ bounded by the length 
$\ell$ of code-words for
the group elements. We compute the map $\psi:G\mapsto \Z_p^\ell$
defined as $x\mapsto
(\phi(x),\phi^{\sigma}(x),\ldots,\phi^{\sigma^{\ell-1}}(x))^T$.
Based on the above discussion, the kernel $M$ of $\psi$
is $\sigma$-invariant. The image $\mbox{Im}(\psi)$ is 
a subspace $V$ of $\Z_p^\ell$. Compute a basis for $V$
by taking a maximal linearly independent set of the 
images of the generators for $G$ under the map $\psi$
and using them replace $\psi$ with the composition
of $\psi$ with the transpose of the matrix whose
columns are the bases elements for $V$. 
This
new map, denoted again by $\psi$, is a surjective homomorphism
from $G$ to $\Z_p^d$ with kernel $M$. Then, by Theorem~\ref{thm:recurs},
after solving the SDLP in the $\psi$-image $\Z_p^d$, 
SDLP$(G,\sigma)$
gets reduced to SDLP$(M,\sigma')$ where $\sigma'$ is the restriction
of a power of $\sigma$ to $M$. 
\end{proof}

\subsection{The SDLP in matrix groups with an inner automorphism}

In this part we prove the following result.

\begin{theorem}
\label{thm:matrix-inner}
Let $G$ be a subgroup of $\GL_d(\F_q)$ where $d$ is a positive integer
and $q$ is a power of a prime. Assume that $G$ is given by a list
of matrices that generate $G$ and that the automorphism $\sigma$ 
is given on the generators. Suppose that $\sigma$ coincides with 
the conjugation action of a matrix $a\in \GL_d(\F_q)$. Then
SDLP$(G,\sigma)$ can be solved by a quantum algorithm in time
polynomial in $d$ and $\log q$.
\end{theorem}

The matrix $a$ that implements the automorphism $\sigma$
does not need to be given, such a matrix is computed by the algorithm.
(It is unique up to the centralizer of $G$.) Note that conjugation by 
$a$ is an inner automorphism of the full matrix group $\GL_d(\F_q)$ (or
just of the matrix group generated by $G$ and $a$), justifying
the title of the subsection.

\begin{proof}
We assume that $q\geq 2d$. (If not, we consider $G$ as a matrix group
over an extension field of $\F_q$ having at least $2d$ elements.)
To find a matrix $a$ with the desired property, we take the linear 
space of matrices $y$  such that $yx_i=\sigma(x_i)y$ for the generators 
$x_i$ of $G$, and choose a random element $a$ of this space. 
Since $q\geq 2d$, by the Schwartz-Zippel lemma \cite{Schwartz,Zippel},
a random element of this matrix space will be with high probability
invertible as it contains at least one by the assumption of
the theorem. Conjugation by $a$ extends $\sigma$ to a linear automorphism
of the full matrix algebra $\mathcal{B}=\mbox{M}_d(\F_q)$ of the $d$ by $d$ matrices.
 We denote this extension also by $\sigma$.

We have $\rho_{(g,1)}(x)=g\sigma(x)$, thus $\rho_{(g,1)}=\mu_g\circ \sigma$, 
where $\mu_g$ denotes the multiplication by $g$ from the left. 
The map $\mu_g$ can also be extended to an invertible linear transformation 
of $\cal B$. Therefore the composition $\rho_{(g,1)}$ has an invertible linear
extension $\Phi$ to $\cal B$.
Also, solving $h=\rho_{(g,1)}^t(1_G)$ is equivalent to solving $h=\Phi^t I_d$. 
The latter is an instance of the well known {\em Orbit Problem} 
introduced by Harrison in \cite{harrison1969lectures}. 
It is the following orbit membership problem. 
Given vectors $a,b$ of a finite dimensional vector space $V$ 
over the field $\F$ and a linear transformation 
$\Phi\in \End_{\F}(V)$, find $t\in \Z_{\geq 0}$, if there exists, 
such that $b=\Phi^t a$. 

Kannan and Lipton in \cite{kannan1986polynomial} gave a polynomial
time solution of the Orbit Problem for the case when $\F$ is the 
field of rationals. Here we need to solve the finite field
case. Kannan and Lipton gave a construction to reduce the
Orbit Problem to the so-called {\em Matrix Power Problem}, 
which is the following. Given square matrices $A$ and $B$ 
over a field $\F$, solve $B=A^t$, see \cite[Theorem 1]{kannan1986polynomial}. 
For completeness we briefly recall (a version of) their construction. 
We compute the subspace $W$ spanned by $\Phi^t a$ ($t=0,1,\ldots$).
This can be done by computing the vectors $a,\Phi a, \ldots, \Phi^{j-1} a$
until 
$\Phi^j a$ becomes linearly dependent of the previous vectors. 
Then $W$ is the subspace with basis $a,\Phi a, \ldots ,\Phi^{j-1} a$.
If $b\not\in W$, then the problem has no solution. Otherwise
$\Phi^t b\in W$ for every $t$. Write the vectors $\Phi^i a$ and $\Phi^i b$
($i=0,\ldots,j-1$) as column vectors in terms of a basis of $W$. Let
$A$ be the matrix of the restriction of $\Phi$ to $W$ in the same basis
and let $C$ resp.~$D$
be the $j$ by $j$ matrices whose columns are $a,\Phi a, \ldots, \Phi^{j-1} a$
and $b,\Phi b, \ldots, \Phi^{j-1} b$, respectively. Then
$b=A^t a$ if and only if $D=A^t C$. Let $B=DC^{-1}$ and 
we need to solve $B=A^t$. If $\Phi$ is invertible then so is $A$.

The invertible case of the matrix power problem over a finite field can be 
solved by Shor's quantum discrete log algorithm. 
\end{proof}

We remark that, using the Jordan blocks of $A$, one could
classically reduce the 
problem to the instances of the discrete logarithm problem 
in the multiplicative group of extensions of $\F$. 
Also,
in practice it might be worth replacing $\cal B$ with the
matrix algebra spanned by the elements of $G$.

Proposition~\ref{prop:topowerofaut} gives the following 
extension.

\begin{corollary}
\label{cor:matrix-powerinner}
Let $G$ be as in Theorem~\ref{thm:matrix-inner}. Let $\sigma$
be an automorphism of $G$. Let $K$ be a positive integer.
We assume that for the divisors $k\leq K$ of the order of $\sigma$,
the action of $\sigma^k$ on the generators for $G$ is also given
and that among those divisors $k$, $\sigma^k$ coincides with
the conjugation action of a matrix. Then SDLP$(G,\sigma)$ can
be solved by a quantum algorithm in time polynomial
in $K,d$ and $\log q$.
\end{corollary}

\subsection{Putting things together}

Our recursion tool (Theorem~\ref{thm:recurs}) can assemble
the results proved in the preceding subsections for various 
special cases of the SDLP to obtain Theorem~\ref{thm:main}.

\begin{proof}[Proof of Theorem~\ref{thm:main}]
Assume that we have the chain of subgroups $M_i$ and 
homomorphisms $\psi_i$ ($i=0,\ldots,k)$ with
properties as in the statement of the theorem. For $i=k$ to
$1$, using Theorem~\ref{thm:recurs}, 
by solving the SDLP in the $\phi_i$-image of $M_i$
we reduce the problem to an instance in $M_{i-1}$. In the
small size case (0), we use brute force. When ${\overline\sigma}_i$ is
of small order (case (1)) or when $\mbox{Im}(\psi_i)$ is solvable
(case (2)),
we use Proposition~\ref{prop:smallorderaut} or Theorem~\ref{thm:solvable},
respectively. In order to facilitate using the oracle for evaluating the powers
of $\sigma$ to evaluate those of ${\overline\sigma}_i$,
we use the pairs $(x,\psi_i(x))$ to encode the elements
of $\mbox{Im}(\psi_i)$, while as labeling we use the labeling
for ${\overline G_i}$. In the matrix group case (4), we use the 
natural encoding by matrices for the image. We compute the order
$o_i$ of ${\overline\sigma}_i$ using the factorization of the order
of $\sigma$ and compute $\sigma^t$ for the smallest few divisors
of $o_i$ and apply the method of Corollary~\ref{cor:matrix-powerinner}.
\end{proof}

\ifnum0=1
\section*{An equivalent form of the group case}\label{equiv}
Now we show that the group case SDLP can be equivalently described as a generalized discrete logarithm problem in any non-commutative group $G$ as follows.

Given $u,v,w\in G$, find all $t$ such that $w=u^{t}v^{-t}$. It will be convenient to state a refinement of the problem, where a normal subgroup $N\lhd G$ with $w,uv^{-1}\in N$
is also given (possibly $N=G$). This is obviously an instance of the
SDLP$(N,\sigma_v)$ for $g=uv^{-1}$ and $h=w$, where 
$\sigma_v$ is the automorphism of $N$ defined as
$\sigma_v:x\mapsto vxv^{-1}$. Conversely, consider an instance of  
SDLP$(N,\sigma)$ for $g,h\in N$ where $N$ is a finite group and
$\sigma$ is an automorphism of $N$. Let $n$ be the order of $\sigma$
and consider the the semidirect product $G=N\rtimes_\sigma \Z/n\Z$ where
$\sigma_{z}=\sigma^z$. Then the SDLP can be cast as an instance of the
problem above by putting $u=(g,1)$, $v=(1_N,0)$, and $w=(h,0)$. 
\fi

\paragraph*{Acknowledgments.}
The research of the second author was supported by the Hungarian 
Ministry of Innovation and Technology NRDI Office within the framework of the Artificial Intelligence National Laboratory Program.

\bibliographystyle{alpha}
\bibliography{myrefs}

\newcommand{\etalchar}[1]{$^{#1}$}
\begin{thebibliography}{CLM{\etalchar{+}}18}

\bibitem[BB99]{babai1999polynomial}
L{\'a}szl{\'o} Babai and Robert Beals.
\newblock A polynomial-time theory of black box groups i.
\newblock {\em London Mathematical Society Lecture Note Series}, pages 30--64,
  1999.

\bibitem[BBS09]{BBS}
L\'{a}szl\'{o} Babai, Robert Beals, and \'{A}kos Seress.
\newblock Polynomial-time theory of matrix groups.
\newblock In {\em Proceedings of the Forty-First Annual ACM Symposium on Theory
  of Computing}, STOC '09, page 55–64, New York, NY, USA, 2009. Association
  for Computing Machinery.

\bibitem[BKPS22]{battarbee2022subexponential}
Christopher Battarbee, Delaram Kahrobaei, Ludovic Perret, and Siamak~F
  Shahandashti.
\newblock A subexponential quantum algorithm for the semidirect discrete
  logarithm problem.
\newblock In {\em NIST Fourth PQC Standardization Conference}, 2022.

\bibitem[BKPS23]{Battarbee2023aspdh}
Christopher Battarbee, Delaram Kahrobaei, Ludovic Perret, and Siamak~F.
  Shahandashti.
\newblock Spdh-sign: Towards efficient, post-quantum group-based signatures.
\newblock In Thomas Johansson and Daniel Smith-Tone, editors, {\em Post-Quantum
  Cryptography}, pages 113--138, Cham, 2023. Springer Nature Switzerland.

\bibitem[BKS22]{battarbee2022semidirect}
Christopher Battarbee, Delaram Kahrobaei, and Siamak~F Shahandashti.
\newblock Semidirect product key exchange: the state of play.
\newblock {\em arXiv preprint arXiv:2202.05178}, 2022.

\bibitem[BS84]{babai1984complexity}
L{\'a}szl{\'o} Babai and Endre Szemer{\'e}di.
\newblock On the complexity of matrix group problems i.
\newblock In {\em 25th Annual Symposium onFoundations of Computer Science,
  1984.}, pages 229--240. IEEE, 1984.

\bibitem[CI14]{childs2014quantum}
Andrew~M Childs and G{\'a}bor Ivanyos.
\newblock Quantum computation of discrete logarithms in semigroups.
\newblock {\em Journal of Mathematical Cryptology}, 8(4):405--416, 2014.

\bibitem[CLM{\etalchar{+}}18]{castryck2018csidh}
Wouter Castryck, Tanja Lange, Chloe Martindale, Lorenz Panny, and Joost Renes.
\newblock Csidh: an efficient post-quantum commutative group action.
\newblock In {\em Advances in Cryptology--ASIACRYPT 2018: 24th International
  Conference on the Theory and Application of Cryptology and Information
  Security, Brisbane, QLD, Australia, December 2--6, 2018, Proceedings, Part
  III 24}, pages 395--427. Springer, 2018.

\bibitem[Cou06]{couveignes2006hard}
Jean-Marc Couveignes.
\newblock Hard homogeneous spaces.
\newblock {\em Cryptology ePrint Archive}, 2006.

\bibitem[Gie95]{giesbrecht1995nearly}
Mark Giesbrecht.
\newblock Nearly optimal algorithms for canonical matrix forms.
\newblock {\em SIAM Journal on Computing}, 24(5):948--969, 1995.

\bibitem[Har69]{harrison1969lectures}
Michael~A Harrison.
\newblock {\em Lectures on linear sequential machines}.
\newblock Academic Press New York, 1969.

\bibitem[HKKS13]{habeeb2013public}
Maggie Habeeb, Delaram Kahrobaei, Charalambos Koupparis, and Vladimir
  Shpilrain.
\newblock Public key exchange using semidirect product of (semi) groups.
\newblock In {\em Applied Cryptography and Network Security: 11th International
  Conference, ACNS 2013, Banff, AB, Canada, June 25-28, 2013. Proceedings 11},
  pages 475--486. Springer, 2013.

\bibitem[IMS01]{ivanyos2001efficient}
G{\'a}bor Ivanyos, Fr{\'e}d{\'e}ric Magniez, and Miklos Santha.
\newblock Efficient quantum algorithms for some instances of the non-abelian
  hidden subgroup problem.
\newblock In {\em Proceedings of the Thirteenth Annual ACM Symposium on
  Parallel Algorithms and Architectures}, pages 263--270, 2001.

\bibitem[KL86]{kannan1986polynomial}
Ravindran Kannan and Richard~J Lipton.
\newblock Polynomial-time algorithm for the orbit problem.
\newblock {\em Journal of the ACM (JACM)}, 33(4):808--821, 1986.

\bibitem[Kup05]{kuperberg2005subexponential}
Greg Kuperberg.
\newblock A subexponential-time quantum algorithm for the dihedral hidden
  subgroup problem.
\newblock {\em SIAM Journal on Computing}, 35(1):170--188, 2005.

\bibitem[Sch79]{Schwartz}
Jacob~T. Schwartz.
\newblock Probabilistic algorithms for verification of polynomial identities.
\newblock In Edward~W. Ng, editor, {\em Symbolic and Algebraic Computation},
  volume~72 of {\em Lecture Notes in Computer Science}, pages 200--215.
  Springer Berlin Heidelberg, 1979.

\bibitem[Sho94]{shor1994algorithms}
Peter~W Shor.
\newblock Algorithms for quantum computation: discrete logarithms and
  factoring.
\newblock In {\em Proceedings 35th annual symposium on foundations of computer
  science}, pages 124--134. Ieee, 1994.

\bibitem[Zip79]{Zippel}
Richard Zippel.
\newblock Probabilistic algorithms for sparse polynomials.
\newblock In Edward~W. Ng, editor, {\em Symbolic and Algebraic Computation},
  volume~72 of {\em LNCS}, pages 216--226. Springer, 1979.

\end{thebibliography}

\appendix

\section*{Appendix: the matrix group case in odd characteristic}

This part is devoted to a sketch of a proof of the following.

\begin{corollary}
Let $\psi:K\rightarrow M_d(\F_q)$ be a representation
of the black-box group $K$ with our without a labeling.
Assume that the automorphism $\sigma$ is given by
a black box to evaluate its powers on elements of $K$ and
that the kernel of $\psi$ is $\sigma$-invariant (e.g.,
when $\psi$ is faithful). Then, $SDLP(\mbox{Im}(\psi),\overline \sigma)$ 
can be solved in quantum polynomial time.
\end{corollary}

\begin{proof}[Proof (sketch)] 

We encode the elements of $G=\mbox{Im}(\psi)$ by 
pairs $(x,\psi(x))$ and labeling $\psi(x)$ so that
we can evaluate powers of $\overline \sigma$ on
elements of the matrix group $G$. We use
the notation $\sigma$ for $\overline \sigma$.
Below we outline how the result of the polynomial time algorithm of 
Babai, Beals and Seress~\cite{BBS} for computing the 
structure of matrix groups over finite fields can be used
to obtain a series of normal subgroups together with representations of 
the factors making Theorem~\ref{thm:main} applicable
to these matrix groups.

Every finite group $G$ has a unique largest solvable
normal subgroup, called the {\em solvable radical} of $G$. 
It is denoted by $\Rad(G)$. The factor group ${\overline G}=G/\Rad(G)$
is trivial if $G$ itself is solvable. Even if $\overline G$ 
is a non-trivial group,
it has no nontrivial abelian subnormal subgroups. (Normal
subgroups of a group are subnormal, and, recursively, normal
subgroups of subnormal subgroups are also subnormal.) It follows
that the minimal subnormal subgroups of $\overline G$ are
non-commutative simple groups. They pairwise commute and the subgroup
$\Soc({\overline G})$ generated by them (called the {\em socle} of
$\overline G$) is the direct product of these simple groups. (It
follows that there are at most $\log \overline G$ simple constituents
of $\Soc({\overline G})$.) The full pre-image of $\Soc({\overline G})$
at the projection $G\rightarrow \overline G$ is denoted by 
$\Soc^*(G)$. The subgroups $\Rad(G)$ and $\Soc^*(G)$ are 
characteristic subgroups of $G$. The group $G$, by conjugation, acts as a permutation
group on the minimal subnormal subgroups of $\overline G$. The
kernel $\Pker(G)$ of this permutation representation, called the {\em permutation
kernel}, is a further characteristic subgroup. There are the following
inclusions between the subgroups introduced above. 
$$1\leq \Rad(G)\leq \Soc^*(G)\leq \Pker(G)\leq G.$$
$\Pker(G)$ acts by conjugation as an automorphism group
on each simple component of $\Soc(\overline G)$. As
the outer automorphism group of any finite simple group
is solvable, the factor group $\Pker(G)/\Soc^*(G)$ is solvable.
The algorithm of Babai, Beals and Seress~\cite{BBS} computes in 
classical randomized polynomial time (generators for) 
the three subgroups above. They also compute a permutation representation
of $G$ with kernel $\Pker(G)$ (this is actually the conjugation
action of $G$ on the simple components of $\Soc({\overline G})$);
a (usually highly intransitive) permutation representation
of $\Pker(G)$ with kernel $\Soc^*(G)$; and, most importantly,
for each of the simple components of $\Soc({\overline G})$,
a sequence of elements of $G$ that generate the component
modulo $\Rad(G)$ together with the images of these under 
an isomorphism with a standard version of the a simple group.
(The generators for each component $S$
are actually generators for a perfect subgroup $S^*$ of $\Soc^*(G)$ 
such that $(S^*/\Rad G)$ is the pre-image of $S$ by the projection
map $G\rightarrow G/\Rad(G)$.) 

We compute refinement of   the sequence 
$1\leq \Rad(G)\leq \Soc^*(G)\leq \Pker(G)\leq G$ 
between $\Rad(G)$ and $\Soc^*(G)$.
To this end we notice that $\sigma$ also permutes the 
simple components of $\Soc{\overline G}$. We take
a $\sigma$-orbit of a single simple component $S$ and
we compute
$S^{**}=\prod_T T^*\Rad(G)$, where the product is taken
over the $\sigma$-orbit of $S$. Let $r$ be the length
of the orbit. Then $\sigma^r$ acts as an automorphism
on each member of the orbit. As the outer automorphism
group of a finite simple group is of size bounded
by a polynomial of the logarithm of the group size,
we obtain that a polynomially small power of $\sigma$
acts as an inner automorphism of $S^{**}/\Rad(G)$.
If $S$ is sporadic, we use the regular representation
of $S$. If $S$ is of Lie type, we take the isomorphism
between $S$ and the standard copy of it computed by
the algorithm of \cite{BBS}. It realizes $S$ as
the quotient group of a matrix group by its center.
We obtain a matrix representation of $S$ by taking the
conjugation action on the matrix algebra spanned by
the elements of this covering group. We do the same
for each $T$ from the orbit (actually, these are isomorphic to
$S$, so the construction made for $S$ can be re-used.)
Finally we obtain a matrix representation of
$S^{**}$ with kernel $\Rad(G)$ on the direct sum
of these representations. Then we proceed with another orbit,
construct the representation of the product of the orbit members
and add to $S^{**}$. This way we obtain a chain of $\sigma$-invariant
normal subgroups between $\Rad(G)$ and $\Soc^*(G)$ together
with the matrix representations of the factors so that
case (4) of Theorem~\ref{thm:main} is applicable to them. For
$G/\Pker(G)$, we use the permutation representation which
can be naturally extended to a matrix representation. As
$\sigma$ also permutes the simple components, the induced
automorphism will be conjugation by a permutation, so case (4) is again applicable.
For $\Pker(G)/\Soc^*{G}$, we use the matrix representation
as a labeling and, by solvability, case (3) is applicable.
Finally, in $\Rad(G)$, again case (3) applies.
\end{proof}

\end{document}